\documentclass[12pt]{amsart}

\usepackage{amsmath,amsthm,amscd,euscript,longtable}
\usepackage{graphicx}
\def \r{\mathbb R}

\def \({\langle}
\def \){\rangle}


\newtheorem{theorem}{Theorem}[section]
\newtheorem{lemma}[theorem]{Lemma}

\newtheorem{proposition}[theorem]{Proposition}
\newtheorem{corollary}[theorem]{Corollary}
\theoremstyle{remark}
\newtheorem{remark}[theorem]{Remark}
\theoremstyle{definition}
\newtheorem{definition}[theorem]{Definition}

\title{Euler elasticae in the plane and the Whitney--Graustein theorem}
\author{S. Avvakumov, O. Karpenkov, A. Sossinsky}
\date{7 November 2012}


\email[Oleg Karpenkov]{karpenkov@tugraz.at}

\email[Alexey Sossinsky]{asossinsky@yandex.ru}

\email[Sergey Avvakumov]{s.avvakumov@gmail.com}

\begin{document}
\input{epsf}

\small{
\begin{abstract}
In this paper, we apply classical energy principles to Euler elasticae, i.e., closed
$\mathcal C^2$ curves in the plane supplied with the Euler functional $U$ (the integral of the square of the curvature along the curve). We study the critical points of  $U$, find the shapes of the curves corresponding to these critical points and show which of the critical points are stable equilibrium points of the energy given by $U$, and which are unstable.
It turns out that the set of stable equilibrium points coincides with the set of minima of
$U$, so that the corresponding shapes of the curves obtained may be regarded as normal forms of Euler elasticae. In this way, we find the solution of the Euler problem
(set in 1744) for plane closed elasticae.  As a by-product, we obtain a ``mechanical" proof of the Whitney--Graustein theorem on the classification of regular curves in the plane up to regular homotopy (in the particular case of $\mathcal C^2$ curves). Besides mathematical theorems, our work includes a computer graphics software which shows, as an animation,  how any plane curve evolves to its normal form under a discretized version of gradient descent along the (discretized) Euler functional.
\end{abstract}
}

\maketitle

\section*{Introduction}

The aim of this paper is to test how the energy functional approach works on
the moduli space of all regular $\mathcal C^2$ curves in the plane $\r^2$
in the case of the Euler functional
$$
U(\gamma) = \int _0^{2\pi}\, \big( \kappa(\gamma(s)) \big)^2 ds,
$$
where $\gamma:\mathbb S^1 \to \r^2$ is a curve of length $2\pi$, $s$ is the arclength parameter, and $\kappa(\gamma(s)$ is the curvature of $\gamma$ at the point $s$.
Thus we view regular homotopy classes of regular plane
curves through  the prism of the Euler functional.

\vspace{1mm}

Let us start with a general remark. Suppose
that one has a configuration space $S$ that is split in several
connected components and the main task is {\it to check if two
elements of $S$ are in the same connected component}.

Then there are two main strategies to do this: combinatorial and
mechanical. The {\it combinatorial strategy} is based on finding
a (hopefully complete) invariant distinguishing connected components.
The {\it mechanical strategy} is as follows: one finds an appropriate
functional (which is sometimes called {\it energy}); then gradient or
descent flow on $S$ is performed along the functional

\noindent and the problem is reduced to comparing local minima of
this flow (these minima are called {\it normal forms}). For smooth
enough moduli spaces $S$ and an appropriate choice of energy
functionals, the gradient flows corresponding to these energies
take almost all (except for a measure zero set) configurations of
$S$ to the corresponding local minima.

\vspace{1mm}

A good example where both strategies are applied is the theory of
knots and links. The combinatorial approach is broadly studied in
classical knot theory, where many beautiful invariants have
been devised (Alexander and Jones polynomials, Vassiliev invariants,
Khovanov homology, etc.). The mechanical strategy, i.e., the
idea of defining energy functionals
for knots is due to H.~K.~Moffat~\cite{Mof}. It was further
developed by W.~Fukuhara~\cite{Fuku}, J.~O'Hara in~\cite{O-H1},
\cite{O-H2}, \cite{O-H3}, \cite{O-H4}, M.~H.~Freedman with various co-authors
in~\cite{Freed2}, \cite{Freed} ~\cite{Freed3}, by
D.~Kim, R.~Kusner in~\cite{Kim}, O.~Karpenkov in~\cite{EKar1},
\cite{EKar2}, etc. Some aspects of the intermediate step between
classical and mechanical approaches is discussed in~\cite{KS}.
In practice, the energy techniques work well, but the mathematical
justification of theoretical questions appears too complicated to be
resolved at the present time. For instance, it is not proven that the
unknot has a unique local minimum with respect to the famous M\"obius
functional invented by O'Hara.

\vspace{1mm}

In this paper, rather than applying the energy techniques to knots
and links, we test the mechanical approach on
regular curves in the plane, which are simpler objects than knots
(to which we intend to return in subsequent publications).
The combinatorial approach to the study of regular curves
in the plane yields the classical Whitney--Graustein theorem, which
provides a classification of curves up to regular homotopy by means
of a simple complete invariant -- the {\it Whitney index} or {\it winding
number}, which is the number of revolutions effected by the tangent
vector to the curve at a mobile point when the mobile point goes once
around the curve. Here we discuss the mechanical
approach to the study of regular plane curves of class $\mathcal C^2$
using the Euler functional $U$
(the integral along the curve of the square of the curvature) and
describe the corresponding normal forms with respect to $U$. It turns
out that each regular homotopy class has a unique normal form
(Theorem~\ref{main}) and the normal form is a complete invariant
of regular homotopy classes. In the case of a nonzero Whitney index,
the normal form is a circle passed once or several times, otherwise it is
Bernoulli's closed $\infty$-shaped elastica (the figure eight curve). These normal
forms are obtained by gradient descent along values of $U$ in the space
of curves. Computer animations of the discretization of this process
(created by the first-named author) are available at [1].

\vspace{1mm}

This paper is organized as follows. We start in
Section~1 with the necessary definitions and preliminaries. In
Section~2, we formulate the main results of this paper (Theorems 1.3
and 2.2) on normal forms of regular curves. In Section 3,
we describe the algorithm on which the animations
are based. Finally, in Section~4, we give technical details and proofs.

In more detail, Section~4, in which critical points of the Euler
functional $U$ are studied, consists of five subsections. In
Subsection~4.1, we discuss the relationship of critical points to
the simple pendulum. We study differentiability questions of
critical curves in Subsection~4.2. In Subsection~4.3, we prove
that all critical curves are either circles passed once or several
times or Bernoulli's closed elastica passed once or several times.
We construct a deformation of critical Bernoulli's closed elastica
(passed more than once) reducing the value of Euler's functional
in Subsection~4.4. Finally, we conclude the proofs of
Theorem~\ref{shapes} and Theorem~\ref{main} in Subsection~4.5.

\section{Preliminaries}

\subsection{Gauss representation of regular curves}

A plane closed curve $\gamma : \mathbb S^1 \to \mathbb R^2$ is
{\it regular} if it possesses a nonzero derivative (tangent
vector) at each point $\gamma(t)$.

For a regular curve $\gamma$, we consider a function
$\alpha:[0,2\pi]\to\r$ such that
$$
\dot\gamma(t)=(\cos\alpha(t),\sin\alpha(t)).
$$
We then say that $\alpha$ is a {\it Gauss representation} of
$\gamma$. (Here and later by $\dot g$ we denote the derivative
${\partial g}/{\partial t}$.) Notice that a pair of curves having
the same Gauss representations coincide after a translation by
some vector.

\vspace{1mm}

In the Gauss representation, we have $\kappa=\dot \alpha$ and hence
the Euler functional is as follows:
$$
U(\alpha)=\int\limits_0^{2\pi}\dot \alpha^2 dt.
$$

\begin{proposition}
A continuous function $\alpha$ is a Gauss representation of some
regular curve if the following conditions hold:
\begin{equation}\label{conditions}
\begin{array}{c}
\int\limits_0^{2\pi}\cos\alpha(t)dt=\int\limits_0^{2\pi}\sin\alpha(t)dt=0;\\
\alpha(0)=\alpha(2\pi).
\end{array}
\end{equation}
\end{proposition}

The proof is obvious. \qed

\subsection{Normal forms of regular curves}

The functional
$$
\int\limits_0^{2\pi}\kappa^2(\gamma(t))dt,
$$
defined for any plane closed curve of class $\mathcal C^2$
is called the {\it Euler functional}; we denote it by $U$. (Notice
that in our previous paper~\cite{KS} we considered a wider class of
functionals. The Euler functional was denoted there by $U_{x^2}$). A
curve supplied with the Euler functional is traditionally (see e.g. [10] 
called an {\it Euler elastica}.

\begin{definition} The {\it normal form} of a regular plane $\mathcal C^2$ curve with
respect to the Euler functional is a plane curve for which the value of $U$
is a local minimum.
\end{definition}

Notice that {\it unstable equilibrium points} for $U$ do not not give
normal forms. Gradient descent reaches such points with zero
probability (such an event  practically never occurs in real life).

\medskip
The main result of the present paper is the following theorem.

\medskip
\begin{theorem}\label{shapes}
{\bf (i)} Any critical regular curve of the Euler functional is
either a circle passed several times or Bernoulli's
$\infty$-shaped closed elastica passed several times:

$$ 
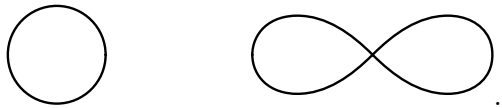
\epsfbox{solution.3}. $$

\bigskip
{\bf (ii)} A Bernoulli closed $\infty$-shaped elastica passed several times is
not stable.

{\bf (iii)} A circle passed once or several times and the Bernoulli
closed elastica passed once are local minima.

\end{theorem}

This result gives a complete solution of the Euler problem for closed elasticas in the
plane. Note that the same result was obtained independently by Yu.Sachkov (see
[19]) by a different, much more laborious method (involving the Pontryagin maximum principle).

\section{The Whitney--Graustein theorem via Euler elasticas}

\begin{definition} The {\it Whitney index} of a regular plane curve $\gamma$ is the
degree of the corresponding Gauss map, taken with a sign (i.e., the number of
clockwise rotations effected by the tangent vector when its origin travels once around the curve). We denote it by
$\omega(\gamma)$.
\end{definition}

For a Gauss representation $\alpha$ of $\gamma$, the following obviously holds:
$$
\omega(\gamma)=\int\limits_\gamma \alpha(t) dt.
$$

\vspace{2mm}

We recall the classical statement of the Whitney--Graustein
theorem: {\it Regular plane curves are classified up to regular
homotopy by their Whitney index, i.e., two regular plane curves
are regularly homotopic if and only if they have the same Whitney
index.}

The following statement immediately follows from
Theorem~\ref{shapes} below and the fact that the normal form
curves appearing in Theorem~\ref{shapes} have Whitney indices
$0, \pm 1, \pm 2, \dots ,$ $\pm n, \dots$ respectively.

\begin{theorem}\label{main}{\bf (On Euler normal forms.)}
Two regular plane curves of class $\mathcal C^2$ are regularly homotopic if and only if
they have the same normal form with respect to the Euler functional
$U$.
\end{theorem}

\begin{remark}
It is not hard to construct different descent flows, we omit the
related analysis. The software [1]
described in the next section gives a
practical, fast (a few seconds for reasonably simple curves), and
very visual method for determining the regular homotopy type of a
plane curve: one simply observes the evolution of the curve until
it reaches its normal form, which determines the regular homotopy type.
In general it would be interesting to have a constructive proof that
such flows always end up in a smooth critical realization. The
second question which we do not touch here is the rigorous construction of
the gradient flow.
\end{remark}

\section{Normal forms via computer animations}

In this section we describe our computer software, which shows how
a given curve evolves to its normal form, and describe the
underlying algorithm.

The software may be downloaded from [1].
It is a user friendly interactive animation and can be viewed on
any PC running with Windows XP (or any later version). Once the
program (exe file) is downloaded and activated, the user simply
draws the required curve on the screen with the mouse (or by means
of the touch pad). The curve begins to evolve, progressively
changing its shape until it arrives to its normal form -- a circle
of radius $\rho,\rho/2, \rho/3, \dots$, or a Bernoulli
$\infty$-shaped curve of length $2\pi \rho$. The evolution of the
curve in the animation to its normal form takes a minute or a few
minutes for moderately complicated curves.

\medskip
The descent algorithm, which is an iterative process, can be
briefly described as follows. First, the computer transforms the
input curve into a closed polygonal line (usually with
self-intersections) with edges $e_i$ of the same (tiny) length and
vertices $v_i$, where $e_i = [v_{i-1},v_i]$, $i=1,\dots N$; here
$N$ is a parameter of our program.

 The discrete version of our functional is

$$
\widehat U := \sum _i \, \tan ^2(\alpha_i / 2),
$$
where $\alpha_i$ is the angle between the continuation of the edge
$e_i$ and the edge $e_{i+1}$.

At each step of the algorithm, for each vertex $v_i$, two
``forces'' (vectors) $s_i$ and $r_i$ are computed. The vector
$s_i$ (the {\it straightening out force}) is calculated according
to the formula
$$
s_i=-C_1\Big (\frac {\partial}{\partial x_i} + \frac
{\partial}{\partial y_i} \Big ) \big (\widehat U \big ),
$$
where $v_i=(x_i,y_i)$ and $C_1$ is a positive constant; here the
partial derivatives ${\partial \widehat U}/{\partial x_i}$ and
${\partial \widehat U}/{\partial y_i}$ are calculated
approximately as finite differences.

The vector $r_i$ (the {\it  resilience force}) is calculated
according to the formula
$$
r_i := C_2( v_{i+1} -  v_{i})(| v_{i+1} - v_{i}|-d_i) +
C_2(v_{i-1} -  v_{i})(|v_i - v_{i-1}|-d_{i-1}),
$$
where $C_2$ is a positive constant, $|v|$ is the Euclidean norm of
the vector $v$, and $d_i$ is the Euclidean distance between the
vertices $v_i$ and $v_{i+1}$ at the initial moment.

Then each of the points $v_i$ is shifted by the vector $s_i +
r_i$, a  new polygonal line is obtained,
 and the algorithm goes on to the next step.

The constants $C_1$ and $C_2$ are parameters of our program and
are chosen so that the lengths of the edges, as well as the total
length of the curve, do not change significantly during the
descent process.

The iteration process continues endlessly (there is no termination
command in the program). However, after a short interval of time,
not more than a few minutes in all our experiments (performed with
$N=100$), the modifications in the shape of the curve become
invisible to the naked eye. We then consider the iteration process
as terminated and regard the shape of the obtained curve as the
{\it normal form} of the input curve. In all our experiments, the
obtained curve was always one of the normal forms predicted by
Theorem 1.3.

Note that this result is not a mathematical theorem, but an
experimental fact. We intend to return to its mathematical
justification in subsequent publications.

\section{Critical points of the Euler functional}

In this section, we study critical points of the Euler functional
$U$ and prove our main result (Theorem~\ref{shapes})
and from it derive the proof of Theorem~\ref{main}.

First, we study the case of twice differentiable Gauss
representations, for which the question is reduced to the equation
of the simple pendulum. Secondly, we show that all critical values
of the Euler functional indeed possess twice differentiable Gauss
representations. The most complicated case is the case of zero
Whitney index: the critical points are the Bernoulli's closed
$\infty$-shaped curves passed several times. We explicitly
construct deformations of these closed elasticae passed several
times that reduce the energy (which proves that these critical
curves are not local minima). In addition, we prove that there are
no simple closed curves satisfying the equation of the simple
pendulum other than circles.

\subsection{The first variation of Euler functional at a critical
elastica}\label{sec1}

Let $\gamma$ be a smooth enough regular plane curve with Gauss
representation $\alpha$.

Consider a variation $\alpha+h\beta$ with a small parameter $h$.
First, since we vary a closed curve, the variation
$\alpha+h\beta$ infinitesimally satisfies
conditions~(\ref{conditions}) above, i.e., the derivative
$\frac{d}{dh}$ of the corresponding integrals equals zero, which is
equivalent to
\begin{equation}\label{e2}
\begin{array}{c}
\int\limits_0^{2\pi}\sin(\alpha(t))\beta(t)dt=
\int\limits_0^{2\pi}\cos(\alpha(t))\beta(t)dt=0,\\
\beta(0)=\beta(2\pi).
\end{array}
\end{equation}
Secondly, for the case in which $\alpha$ is a critical elastica,
the first variation is zero. Hence
\begin{equation}\label{var1}
\int\limits_0^{2\pi}\dot \alpha\dot \beta dt=0,
\end{equation}
which is equivalent to
\begin{equation}\label{var2}
\int\limits_0^{2\pi}\alpha\ddot \beta dt=0.
\end{equation}
Now we have an equation for critical points even if the Gauss
representation is only continuous.

\vspace{1mm}

There is another equivalent representation of the equation for the
first variation
$$
\int\limits_0^{2\pi}\ddot\alpha \beta dt=0.
$$
The last equation holds for any variation $\beta$ satisfying
Equations~(\ref{e2}), therefore, we have the following statement.

\begin{corollary}\label{pendulum}
Let $\gamma$ have a twice differentiable Gauss representation
$\alpha$. Suppose that $\gamma$ is critical for the Euler
functional. Then there exist constants $C_1$ and $C_2$ such that
$\alpha$ satisfies
\begin{equation}\label{e1}
\ddot \alpha=C_1\cos \alpha+C_2\sin\alpha.
\end{equation}
\end{corollary}

\begin{remark}\label{discussion}
If $C_1=C_2=0$ then we have the equation of circles: $\ddot
\alpha=0$. If at least one of the constants $C_1$ and $C_2$ is not
equal to zero, then after an appropriate Euclidean transformation,
we obtain the equation for the simple pendulum
$$
\ddot \alpha+\omega^2\sin\alpha=0
$$
for some nonnegative constant $\omega$.
\end{remark}

\subsection{Smoothness of critical elasticae}

Let us prove the smoothness of critical elasticae.
The proof of this assertion is traditionally missing in the
literature, it is usually supposed that the curve is smooth
enough.

\begin{proposition}\label{cont}
The Gauss representation of a critical elastica is twice
differentiable at any point.
\end{proposition}

We prove the assertion of the proposition in three steps. First, we
show that the Gauss representation is continuous. Secondly, we
show that it is continuously differentiable. Finally, we prove that it is
twice differentiable.

\subsubsection{Continuity of the Gauss representation}

\begin{lemma}\label{gicont}
The Gauss representation of critical elasticas is continuous at each
point.
\end{lemma}

\begin{remark}
If $\alpha$ is not differentiable, we can understand
$\dot\alpha$ in the generalized way, as the difference between the
corresponding upper and lower bounds. Similarly, we can consider
the function $\alpha$ as the $L^2$-limit of smooth functions and
calculate the derivative $\dot\alpha$ as the $L^2$-limit of the
derivatives of these functions.
\end{remark}

\begin{proof}
If the assertion of Lemma~\ref{gicont} is not true, then there
exists a constant $C$ such that for any $\varepsilon>0$ there
is a $t_0$ such that
$$
\int\limits_{t_0}^{t_0+\varepsilon} |\dot\alpha|dt>C
$$
(the integral in the left-hand side of the inequality can be infinite). Hence
$$
\int\limits_{t_0}^{t_0+\varepsilon} \dot\alpha^2
dt>\frac{C^2}{\varepsilon}.
$$
Therefore, the Euler functional is infinite for this curve. We
come to a contradiction.
\end{proof}

\subsubsection{Variations $\hat\beta_{a,\varepsilon,b,\xi}$}

Denote by $\delta(x)$ the generalized Dirac $\delta$-function.
By definition, let us put
$$
\begin{array}{l}
\displaystyle
\hat\beta_{a,\varepsilon,b,\xi}(x)=\frac{1}{b-a}\int\limits_{0}^{x}\int\limits_0^y
\frac{\delta(z-b-\xi)-\delta(z-b)}{\xi}-
\frac{\delta(z-a-\varepsilon)-\delta(z-a)}{\varepsilon}dzdy
\end{array}
$$
In addition, we extend this function as follows
$$
\begin{array}{l}
\hat\beta_{a,0,b,\xi}(x)=\frac{1}{b-a}\int\limits_{0}^{x}
\Big(\int\limits_0^y \frac{\delta(z-b-\xi)-\delta(z-b)}{\xi}dz
-\delta(y-a)\Big)dy;\\
\hat\beta_{a,0,b,0}(x)=\frac{1}{b-a}\int\limits_{0}^{x}
\big(\delta(y-b)-\delta(y-a)\big)dy;\\
\end{array}
$$

Let
$$
\beta_{a,\varepsilon,b,\xi}=\hat\beta_{a,\varepsilon,b,\xi}-h_\alpha(\hat\beta_{a,\varepsilon,b,\xi}),
$$
where $h_\alpha$ is the orthogonal projection from the space of
$L^2$-functions to the vector space spanned by the functions
$\cos(\alpha(x))$, $\sin(\alpha(x))$ and constant functions.

\begin{lemma}
The function $\beta_{a,\varepsilon,b,\xi}$ is continuous in the
$L^2$-norm at all points in which we have already defined it. \qed
\end{lemma}

The proof is straightforward, so we omit it.

\vspace{2mm}

Finally, we put
$\hat\beta_{a,0,a,0}(x)=\delta(x-a).$

\vspace{2mm}

\subsubsection{Orthogonal basis in $H_\alpha$}

Denote by $H_\alpha$ the linear span of the functions
$\cos\alpha$, $\sin\alpha$ and constant functions. For an
arbitrary choice of $\alpha$, let us fix an orthogonal basis:
$$
\begin{array}{ccl}
e_{\alpha,1}&=&\cos\alpha;\\
e_{\alpha,2}&=&\sin\alpha+\frac{\int\limits_0^{2\pi}\sin\alpha(x)\cos\alpha(x)dx}
{\int\limits_0^{2\pi}\sin^2\alpha(x)dx}\cos\alpha;
\\
e_{\alpha,3}&=&1+
\frac{\int\limits_0^{2\pi}\sin\alpha(x)\cos\alpha(x)dx\int\limits_0^{2\pi}\cos\alpha(x)dx-
\int\limits_0^{2\pi}\cos^2\alpha(x)dx\int\limits_0^{2\pi}\sin\alpha(x)dx}
{\int\limits_0^{2\pi}\sin^2\alpha(x)dx\int\limits_0^{2\pi}\cos^2\alpha(x)dx-
\Big(\int\limits_0^{2\pi}\sin\alpha(x)\cos\alpha(x)dx\Big)^2}
\sin\alpha+
\\
&&
\frac{\int\limits_0^{2\pi}\sin\alpha(x)\cos\alpha(x)dx\int\limits_0^{2\pi}\sin\alpha(x)dx-
\int\limits_0^{2\pi}\sin^2\alpha(x)dx\int\limits_0^{2\pi}\cos\alpha(x)dx}
{\int\limits_0^{2\pi}\sin^2\alpha(x)dx\int\limits_0^{2\pi}\cos^2\alpha(x)dx-
\Big(\int\limits_0^{2\pi}\sin\alpha(x)\cos\alpha(x)dx\Big)^2}
\cos\alpha.\\
\end{array}
$$
The Gauss representation for closed curves is not constant,
therefore the denominator of the second coefficient for
$e_{\alpha,2}$ is nonzero. By the Cauchy-Schwarz inequality for
$L^2$-functions, the denominators of the coefficients of
$e_{\alpha,3}$ are also nonzero. The following statement is
straightforward.

\begin{lemma}
The functions $e_{\alpha,1}$, $e_{\alpha,2}$, and $e_{\alpha,3}$
are smooth bounded nonzero functions. \qed
\end{lemma}

As a corollary we have the following.

\begin{corollary}\label{L2Cor}
The function $h_\alpha(\hat\beta_{t,\varepsilon,u,\xi})(x)$ is a
continuous function that $L^2$-continuously depends on the
parameters $(t,\varepsilon,u,\xi)$ for an arbitrary 4-tuple of
parameters satisfying $0\le t,u \le 2\pi$ $($including
$(t,0,t,0)$$)$. \qed
\end{corollary}

\subsubsection {Continuous differentiability of the Gauss representation}

\begin{lemma}\label{diff}
The Gauss representation of a critical elastica is continuously
differentiable at each point.
\end{lemma}

\begin{proof}
Let us fix $t\ne u$, and $\xi$, and vary $\varepsilon$. From
Equation~\ref{var2} it follows that
$$
\frac{\alpha(t+\varepsilon)-\alpha(t)}{\varepsilon}=
\frac{\alpha(u+\xi)-\alpha (u)}{\xi}+(t-u)
\int\limits_{0}^{2\pi}\alpha(x)(h_\alpha(\hat\beta_{t,\varepsilon,u,\xi})(x))''dx.
$$
The first summand of the left part does not depend on
$\varepsilon$. The second summand is a continuous function in the
$\varepsilon$ variable, since $\alpha$ is continuous. Therefore,
the limit of the expression in the right-hand side of the
equality exists, i.e., the derivative $\dot\alpha$ exists at $t$.

\vspace{2mm}

Let us fix $t$, and vary $\varepsilon$. From Equation~\ref{var2}
it follows that
$$
\Big|\frac{\alpha(t+\xi)-\alpha(t)}{\xi}-
\frac{\alpha(t+\varepsilon+\xi)-\alpha(t+\varepsilon)}{\xi}\Big|=
\varepsilon
\Big|\int\limits_{0}^{2\pi}\alpha(x)(h_\alpha(\hat\beta_{t,\xi,t+\varepsilon,t+\varepsilon+\xi})(x))''dx\Big|<C\varepsilon,
$$
where $C$ does not depend on $\varepsilon$ and $\xi$. Therefore,
$$
|\dot\alpha(t)-\dot\alpha(t+\varepsilon)|<C\varepsilon.
$$
Thus, the function $\dot\alpha$ is continuous at $t$.


\end{proof}

\subsubsection{Conclusion of the proof of Proposition~\ref{cont}}

From Proposition~\ref{diff} it follows that the Gauss representation
of a critical elastica is continuously differentiable at each
point.

From Equation~\ref{var1} it follows that
$$
\frac{\dot\alpha(t+\varepsilon)-\dot\alpha(t)}{\varepsilon}=
\int\limits_{0}^{2\pi}\dot\alpha(x)(h_\alpha(\hat\beta_{t,0,t+\varepsilon,0})(x))'dx.
$$
Therefore,
$$
\ddot\alpha(t)=\lim\limits_{\varepsilon\to 0}
\int\limits_{0}^{2\pi}\dot\alpha(x)(h_\alpha(\hat\beta_{t,0,t+\varepsilon,0})(x))'dx.
$$
Corollary~\ref{L2Cor} implies that the function
$h_\alpha(\hat\beta_{t,\varepsilon,u,\xi})(x)$ $L^2$-continuously
depends on $\varepsilon$. Therefore, the limit exists. Hence the
Gauss representation is twice differentiable. \qed

\subsection{Uniqueness of $\infty$-shaped normal forms passed once}

In this section we briefly analyze the critical closed elasticae
whose Gauss representations satisfy the equation
$$
\ddot\alpha+\sin\alpha=0.
$$
It is clear that the elasticae whose Gauss representations
correspond to the motions of the pendulum that makes complete
turns are not bounded. If the pendulum does not make a complete
turn, then it is possible to get a closed elastica. All such
curves are homotopic to the figure ``$\infty$'' curve. We call them
{\it $\infty$-shaped normal forms}. The main statement about
$\infty$-shaped normal forms is as follows.

\begin{proposition}\label{uniqueness}
All $\infty$-shaped normal forms are homothetic to each other.
\end{proposition}

In the proof of Proposition~\ref{uniqueness}, we essentially used
the following general statement.

\begin{proposition}\label{curv}{\bf (On 2-germ similarity.)}
Consider two $\mathcal C^2$-curves $\gamma_1:[0,T_1]\to\r^2$ and
$\gamma_2:[0,T_2]\to\r^2$ with Gauss representations $\alpha_1$
and $\alpha_2$ respectively.

Suppose that the following conditions hold:

$\bullet$ the curves are convex;

$\bullet$ the curves are inscribed in the same angle centered at $O$;

$\bullet$ the curves have the same starting point:
$\gamma_1(0)=\gamma_2(0)$.

Then there exists a pair of points $(t_1,t_2)$ simultaneously satisfying
$$
\alpha_1(t_1)=\alpha_2(t_2) \quad \hbox{and} \quad \dot
\alpha_1(t_1)=\dot\alpha_2(t_2).
$$
\end{proposition}

We start with the following two lemmas.

\begin{lemma}\label{zzz}
Consider two $\mathcal C^2$-curves $\gamma_1:[0,T_1]\to\r^2$ and
$\gamma_2:[0,T_2]\to\r^2$ with Gauss representations $\alpha_1$
and $\alpha_2$ respectively.

Suppose that the following conditions hold:

$\bullet$ the curves are convex;

$\bullet$ the curves are inscribed in the same angle centered at $O$;

$\bullet$ the curves have the same starting point:
$\gamma_1(0)=\gamma_2(0)$.

$\bullet$ for any pair of points $(t_1,t_2)$ satisfying
$\alpha_1(t_1)=\alpha_2(t_2)$ the inequality
$\dot\alpha_1(t_1)>\dot\alpha_2(t_2)$ holds.

Then the point $\gamma_1(T_1)$ is contained in the interior of the
segment with endpoints $O$ and $\gamma_2(T_2)$ $($see
Figure~$1$,
left$)$.
\end{lemma}

\begin{figure}
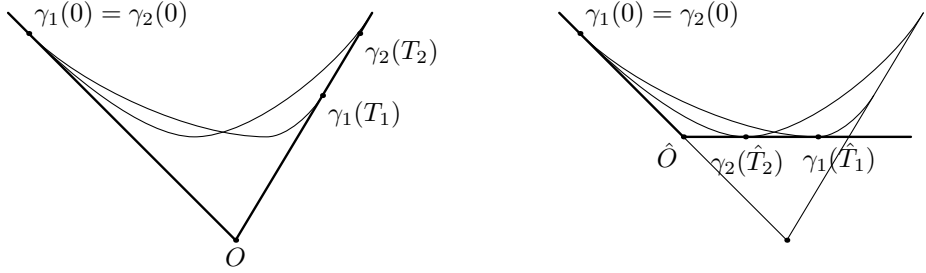

$$
\epsfbox{picture.1} \qquad \qquad\epsfbox{picture.2}
$$
\caption{Two curves inscribed at the same angle (left) and their
common tangent line (right).}\label{pictureA}
\end{figure}

\begin{proof}
Without loss of generality, we assume that the angle in which both
curves are inscribed is defined by the rays $y=\tan\alpha_0, x>0$
and $y=-\tan\alpha_0, x<0$. Let the starting point be on the left
ray. So both $\alpha_1$ and $\alpha_2$ are increasing functions
from $-\alpha_0$ to $\alpha_0$.

The condition that {\it for any pair of points $(t_1,t_2)$ satisfying
$\alpha_1(t_1)=\alpha_2(t_2)$ the inequality
$\dot\alpha_1(t_1)>\dot\alpha_2(t_2)$ holds} implies that
$$
\int\limits_{0}^{T_1}\cos\alpha_1(t)dt<\int\limits_{0}^{T_2}\cos\alpha_2(t)dt.
$$
Hence the $x$-coordinate of $\alpha_1(T_1)$ is smaller than the
$x$-coordinate of $\alpha_2(T_2)$. Therefore, the point
$\gamma_1(T_1)$ is contained in the interior of the segment with
endpoints $O$ and $\gamma_2(T_2)$.
\end{proof}

\begin{lemma}\label{zzzzz}
There are no $C^2$-curves satisfying all the conditions of
Lemma~\ref{zzz}.
\end{lemma}

\begin{proof}
We prove this by reductio ad absurdum. Suppose that such curves
$\gamma_1$ and $\gamma_2$ exist.

On the one hand, by Lemma~\ref{zzz} the point $\gamma_1(T_1)$ is
contained in the interior of the segment with endpoints $O$ and
$\gamma_2(T_2)$. On the other hand, the condition
$\dot\alpha_1(0)>\dot\alpha_2(0)$ implies that there are some
points $\gamma_2(t)$ (with small parameter $t$) that lie in the
same connected component in the complement of the angle to the
curve $\gamma_1$. Hence, there is a point where $\gamma_1$ crosses
$\gamma_2$.

Hence there exists a line $l$ which is tangent to both curves
$\gamma_1$ and $\gamma_2$. Suppose that this happens at times
$\hat T_1$ and $\hat T_2$ respectively. Let $l$ cross the left ray $r$
at the point $\hat O$ (see Figure~$1$,
right). Notice that in these settings, the point $\gamma_2(\hat T_2)$ is
contained in the interior of the segment with endpoints $\hat O$
and $\gamma_1(\hat T_1)$. We come to a contradiction with
Lemma~\ref{zzz} for the curves $\gamma_1:[0,\hat T_1]\to\r^2$ and
$\gamma_2:[0,\hat T_2]\to\r^2$, which are both inscribed in the
angle with vertex $\hat O$.
\end{proof}

{\noindent {\it Proof of Proposition~\ref{curv}.}  From
Lemma~\ref{zzzzz}, it follows that if all the conditions are
satisfied, then there are two pairs $(t_1,t_2)$ and $(t_3,t_4)$
such that
$$
\left\{
\begin{array}{l}
\alpha_1(t_1)=\alpha_2(t_2)\\
\dot\alpha_1(t_1)>\dot\alpha_2(t_2)
\end{array}
\right.
\qquad \hbox{and} \qquad
\left\{
\begin{array}{l}
\alpha_1(t_3)=\alpha_2(t_4)\\
\dot\alpha_1(t_3)<\dot\alpha_2(t_4)
\end{array}
\right..
$$
Hence, for continuity reasons, there exists a pair $(t_5,t_6)$ for which
we have both
$$
\alpha_1(t_1)=\alpha_2(t_2)\qquad \hbox{and} \qquad
\dot\alpha_1(t_5)=\dot\alpha_2(t_6).
$$
\qed}

\begin{lemma}\label{time}
The duration of the period increases with the amplitude.
\end{lemma}

\begin{proof}
This follows directly from the fact that the period $T$ with
amplitude $\alpha_0$ is calculated by the formula:
$$
T=4\omega K\Big(\sin^2\frac{\alpha_0}{2}\Big)
=4\omega\int\limits_{0}^{\pi/2}\frac{d\theta}{\sqrt{1-\sin^2\frac{\alpha_0}{2}\sin^2\theta}}
$$
(here $K(t)$ is the complete elliptic integral of the first type).
The function under the integral sign increases when $\alpha_0$
increases.
\end{proof}

{\noindent {\it Proof of Proposition~\ref{uniqueness}.} We argue
by reduction ad absurdum. Suppose that there are two different
closed $\infty$-shaped solutions of the equation
$$
\ddot \alpha +\sin\alpha=0
$$
with amplitudes $\alpha_1$ and $\alpha_2$. Let
$\alpha_2>\alpha_1$. Consider the two curves $\gamma_1$ and $\gamma_2$
corresponding to one fourth of the $\infty$-shaped curves,
starting from the point with vertical tangent vector (see
Figure 2).

\begin{figure}
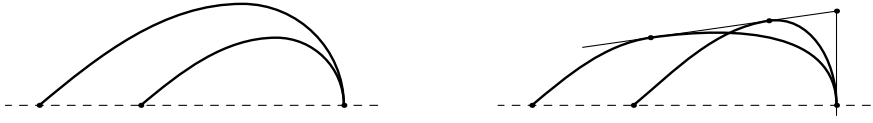

$$
\epsfbox{picture.3} \qquad \qquad\epsfbox{picture.4}
$$
\caption{The curves $\gamma_1$ and $\gamma_2$ either do not
intersect (on the left) or intersect (on the right). In the last
case some of their parts are inscribed in the same
angle.}\label{pictureB}
\end{figure}

Let us prove that $\gamma_1$ and $\gamma_2$ do not intersect at
interior points. Suppose that the converse is true. If the
curves $\gamma_1$ and $\gamma_2$ intersect, then there is some
line which is tangent to both of these curves. Therefore, there
are some parts of them that are inscribed in the same angle (see
Figure~$2$,
 right). On the one hand, by
Proposition~\ref{curv} there exists a pair of points $(t_1,t_2)$
satisfying simultaneously
$$
\alpha_1(t_1)=\alpha_2(t_2) \quad \hbox{and} \quad \dot
\alpha_1=\dot \alpha_1(t_1)=\dot\alpha_2(t_2).
$$
On the other hand, any solution of the pendulum equation ($\ddot
\alpha +\sin\alpha=0$) satisfies
$$
\dot \alpha=\pm\sqrt{\cos\alpha+C}.
$$
Combining these, we obtain the following:
$$
\sqrt{\cos\alpha_2(t_1)+C_1}=\sqrt{\cos\alpha_1(t_1)+C_1}=\dot
\alpha_1(t_1)=\dot\alpha_2(t_2)=\sqrt{\cos\alpha_2(t_1)+C_2},
$$
which implies $C_1=C_2$. Hence the curves $\gamma_1$ and
$\gamma_2$ coincide. We come to a contradiction. Therefore, the
curves $\gamma_1$ and $\gamma_2$ do not intersect at inner
points.

\vspace{2mm}

Since the curves do not intersect and
$\dot \alpha_1(0)<\dot \alpha_2(0),$
the curve $\gamma_1$ lies above the curve $\gamma_2$. On the one
hand, since $\gamma_1$ is convex, and both ends of both curves are
on the $OX$ coordinate axis, the length of $\gamma_1$ is not less than
the length of $\gamma_2$. On the other hand, Lemma~\ref{time}
implies that the length of $\gamma_1$ is less than the length of
$\gamma_2$. We come to a contradiction.
\qed

\begin{remark}\label{BL}
The unique $\infty$-shaped curve is called the {\it Bernoulli
closed elastica}. Its arc length representation is closely related
to Bernoulli's lemniscate. In~\cite{Lev} and~\cite{Srid} there are
good historical overviews about elasticae in general, which contain, in
particular, good descriptions of this curve.
\end{remark}

\subsection{Unstable critical elasticae for the Euler functional}

In this subsection, we prove the following statement.

\begin{proposition}\label{unst}
Bernoulli's closed elastica passed several times is not
stable.
\end{proposition}

{\bf Construction of $\Gamma_\varepsilon$}. Consider the loop of
the closed elastica with center $O_1$ at the origin with
nonpositive first coordinate. Construct another loop of
Bernoulli's closed elastica centrally symmetric to the first one
and tangent to the first loop at the point $A$ with
first coordinate  $-\varepsilon/2$ and positive second
coordinate (see Figure~$3$,
left). Denote the center of the second loop by $O_2$. Find the
point $O_3$ on the $OY$-axis such that the line $O_2O_3$ touches
the upper branch of the elastica at $O_2$ and connect $O_3$
with $O_2$ by a line segment (Figure~$3$,
right).

\begin{figure}[h]
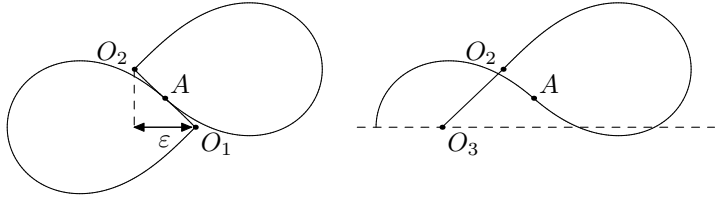

$$
\epsfbox{example.1}\quad\epsfbox{example.2}
$$
\caption{Preliminary steps to construct
$\Gamma_\varepsilon$.}\label{example}
\end{figure}

\medskip
Finally, add the symmetric picture about the $OY$-axis and add one
more Bernoulli closed elastica loop centered at $O_3$, as on
Figure $4$,
 left. Denote the symmetric point to
$O_2$ by $O_4$.

\begin{figure}[h]
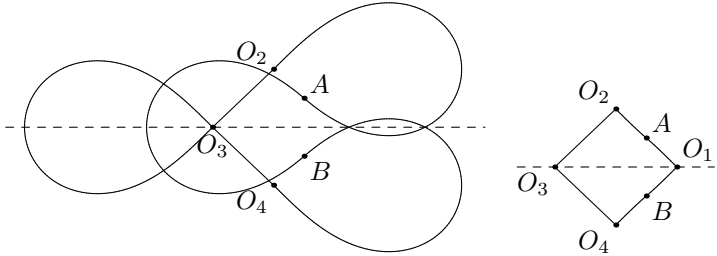

$$
\epsfbox{example.3}\quad\epsfbox{example.4}
$$
\caption{The embedding $\Gamma_\varepsilon$ and its difference
from the union of four Bernoulli's closed elastica
loops.}\label{example2}
\end{figure}

{\it Remark.} The immersion $\Gamma_\varepsilon$ is not
$C^{\infty}$-smooth but only $C^1$-smooth. We omit standard smoothing
procedures here.

\vspace{2mm}

{\noindent {\it Proof of Proposition~\ref{unst}.} Let us briefly
estimate the difference in lengths and curvatures between
$\Gamma_\varepsilon$ and the double Bernoulli closed elastica
(i.e. the union of the four loops of Bernoulli's closed elastica).}

First, we compare the lengths. If we replace two segments $O_2O_3$
and $O_3O_4$ in $\Gamma_\varepsilon$ by the curves $O_1AO_2$ and
$O_3BO_4$ constructed via the corresponding parts of Bernoulli's
closed elastica (Figure~$4$,
right), then the length
will be exactly equal to twice the length of the closed
elastica. Notice that the angle $O_2O_1O_3$ is less than the angle
$O_2O_3O_1$, since all the absolute values of the derivatives at
points of the closed elastica are greater than
$\tan(O_2O_3O_1)$ almost everywhere. Thus the length of $O_1O_2$ is
greater than the length of $O_3O_2$. Hence the length of $O_3O_2$ is less
than the length of the curve $O_1AO_2$. For the same reason, the
length of $O_3O_4$ is less than the length of the curve
$O_1BO_4$. Therefore, the length of $\Gamma_\varepsilon$ is
smaller than the length of the double Bernoulli closed
elastica. The curvature of the Bernoulli closed elastica part of
$\Gamma_\varepsilon$ (the bold $\infty$-curve in Figure~5)
coincides with the curvatures of the
corresponding points on the double  Bernoulli's closed elastica.
The curvature at points on the additional segments are equal to zero.
Since the lengths are smaller and the curvatures at the
corresponding points are not greater, for the total Euler
functional we have
$$
U(\Gamma_\varepsilon)<U(\Gamma_0).
$$
Hence, the double  Bernoulli's closed elastica is a saddle point
of the configuration space of all immersions of $0$ index.

\begin{figure}[h]
$$
\epsfbox{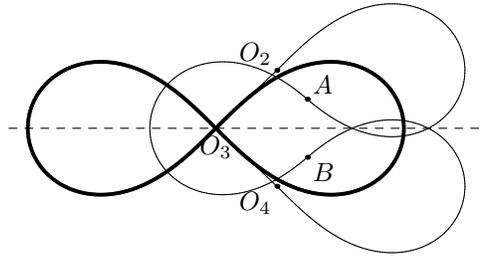}
$$
\caption{The immersion $\Gamma_\varepsilon$ for n-tuple
Bernoulli's closed elasticae.}\label{example3}
\end{figure}

In the case of an $n$-tuple Bernoulli closed elastica for $n>2$, we add
$(n-2)$-tuple  Bernoulli closed elasticae to all the immersions
of the above construction of the energy reducing deformation for the
double Bernoulli closed elastica, as it is shown on
Figure~$5$.
The $(n-2)$-tuple Bernoulli closed
elastica remains unchanged during the whole deformation. \qed

\subsection{Proof of Theorem~\ref{shapes} and Theorem~\ref{main}}
In this subsection, we conclude the proofs of the main theorems.

\subsubsection{Proof of Theorem~\ref{shapes}} {\bf (i)}
By Proposition 4.3,
the Gauss representation of any critical
elastica is twice differentiable at any point. By
Corollary~\ref{pendulum} (see also Remark~\ref{discussion}) all
critical elasticas with twice differentiable Gauss representation
are either circles or satisfy the equation of the simple pendulum.
By Proposition  4.10,
all critical elasticae whose
Gauss representation satisfy the equation of the simple pendulum
are homothetic to Bernoulli's closed elastica (see
Remark~\ref{BL}). The first item is proved.

\smallskip
{\bf (ii)} We have proved this item in Proposition~\ref{unst}.

\smallskip

{\bf (iii)} The first two items of this theorem imply that each
connected component has a unique stable critical point.  Hence
these unique critical points are local minima. Recall that in the
case of a nonzero Whitney index, the corresponding curve
is a circle passed a number of
times equal to the Whitney index (the orientation is determined by
the sign of the index). If the Whitney index is zero, then the
critical elastica is the Bernoulli closed $\infty$-shaped elastica. This
concludes the proof of Theorem~\ref{shapes}. \qed

\subsubsection{Proof of Theorem~\ref{main}} Theorem~\ref{shapes}
implies that each regular homotopy class of regular curves
contains a unique normal form. \qed

\medskip
\noindent {\bf Acknowledgements.} We are grateful to Mikhail
Zelikin for an illuminating informal talk about Euler elasticae,
to Alexander Demidov for pointing out the references [18] and
[19], and to Robert Goldstein for a useful discussion of the
present paper. Oleg Karpenkov is supported by the Austrian Science
Fund (FWF), grant M~1273-N18. Alexey Sossinsky is partially
supported by the RFBR-CNRS-a grant \#10-01-93111 and the RFBR
grant \# 12-01-00748-a.


\begin{thebibliography}{99}


\bibitem{Avva} S.~Avvakumov, software available at  http://www.mccme.ru/knotenergy

\bibitem{Freed3} S.~Bryson, M.~H.~Freedman, Z.-X.~He, and Z.~Wang, {\it M\"obius invariance of knot energy},
Bull. Amer. Math. Soc. (N.S.),  vol.~28 (1993), no.~1,
pp.~99--103.

\bibitem{Eul} L.~Euler, {\it Methodus inveniendi lineas cyrvas maximi minimive proprietate gaudentes, sive Solutio problematis isoperimitrici latissimo sensu accepti,} Lausanne, Gen\`eve, 1744.

\bibitem{Freed2} M.~H.~Freedman, Z.-X.~He, {\it Links of tori
and the energy of incompressible flows}, Topology vol.~30 (1991),
no.~2, pp.~283--287.

\bibitem{Freed}
M.~H.~Freedman, Z.-X.~He, and Z.~Wang, {\it M\"obius energy of
knots and unknots}, Ann. of Math. (2), vol.~139 (1994), no.~1,
pp.~1--50.

\bibitem{Fuku}
W.~Fukuhara, {\it Energy of a knot}, The f\^ete of topology,
Academic Press, (1988), pp~443--451.

\bibitem{KS}
O.~Karpenkov, A.~Sossinsky, {\it Energies of knot diagrams},
Russian J. of Math. Phys., vol.~18(2011), no.~3, pp.~306--317.

\bibitem{EKar1}
O.~Karpenkov, {\it Energy of a knot: variational principles},
Russian J. of Math. Phys. vol.~9(2002), no.~3, pp.~275--287.

\bibitem{EKar2}
O.~Karpenkov, {\it Energy of a knot: some new aspects} Fundamental
Mathematics Today, Nezavis. Mosk. Univ., Moscow (2003),
pp.~214--223.

\bibitem{EKar3}
O.~Karpenkov, {\it The M\"obius energy of graphs}, Math. Notes,
vol.~79(2006), no.~1-2, pp.~134--138.

\bibitem{Kim}
D.~Kim, R.~Kusner, {\it Torus Knots Extremizing the M\"obius
Energy}, Experimental Mathematics, vol.~2(1993), no.~1, pp.~1--9.

\bibitem{Lev}
R.~Levien, {\it The elastica: a mathematical history}, EECS
Department, University of California, Berkeley, Technical Report
No. UCB/EECS-2008-103, 2008.\\
http://www.eecs.berkeley.edu/Pubs/TechRpts/2008/EECS-2008-103.html

\bibitem{Mof}
H.~K.~Moffat, {\it The degree of knottedness of tangled vortex
lines}, J. Fluid Mech. vol.~35(1969), pp.~117--129.

\bibitem{O-H1}
J.~O'Hara, {\it Energy of a knot}, Topology, vol.~30(1991), no.~2,
pp.~241--247.

\bibitem{O-H2}
J.~O'Hara, {\it Family of energy functionals of knots}, Topology
Appl. vol~48(1992), no.~2, pp.~147--161.

\bibitem{O-H3}
J.~O'Hara, {\it Energy functionals of knots II}, Topology Appl.
vol.~56(1994), no.~1, pp.~45--61.

\bibitem{O-H4}
J.~O'Hara, {\it Energy of Knots and Conformal Geometry}, K~\&~E
Series on Knots and Everything -- Vol.~33, World Scientific, 2003,
288~p.

\bibitem{Os-Zel} Yu.S.Osipov, M.I.Zelikin, {\it Higher-order Euler elastics and elastic hulls}, Russ.J. Math.Phys.
{\bf 19}, no. 2, 163-172 (2012)

\bibitem{Sach}
Yu.L.Sachkov, {\it Closed Euler elasticae,} Progam Syst.Inst. Preprint, 2011.

\bibitem{ABS1}
A.~B.~Sossinsky, {\it Mechanical Normal Forms of Knots and Flat
Knots}, Russ. J. Math. Phys. vol.~18, no.~2(2011).

\bibitem{Srid}
R.~Sridharan, {\it Physics to mathematics: from lintearia to
lemniscate - I}, Resonance(2004), pp.~21--29.

\end{thebibliography}
\end{document}